\documentclass{svjour3}
\def\makeheadbox{}
\usepackage{mathptmx}
\usepackage{amsmath,amssymb}
\usepackage{mathrsfs}
\usepackage{microtype}
\usepackage{upref}
\let\mathcal\mathscr
\usepackage[mathscr]{eucal}

\smartqed

\newcommand{\cprime}{\/{\mathsurround=0pt$'$}}

\allowdisplaybreaks[4]

\newcommand*{\pd}[2]{\mathchoice{\frac{\partial#1}{\partial#2}}
  {\partial#1/\partial#2}{\partial#1/\partial#2}
  {\partial#1/\partial#2}}

\newcommand{\envert}[2][\right]{\relax
  \ifx#1\right\relax \left\lvert\else#1\lvert\fi#2#1\rvert}
\let\abs=\envert

\newcommand{\doubell}{\mathcal{L}}

\let\kappa\varkappa
\let\phi\varphi

\newcommand{\hd}{\bar{d}}

\DeclareMathOperator{\CDiff}{\mathcal{C}Dif{}f}
\DeclareMathOperator{\cl}{cl}
\DeclareMathOperator{\Hom}{Hom}
\DeclareMathOperator{\Sym}{Sym}
\DeclareMathOperator{\CoSym}{CoSym}
\DeclareMathOperator{\im}{im}

\DeclareFontFamily{OT1}{wncyi}{}
\DeclareFontShape{OT1}{wncyi}{m}{it}{
   <5> <6> <7> <8> <9> gen * wncyi
   <10> <10.95> <12> <14.4> <17.28> <20.74> <24.88> wncyi10
  }{}
\DeclareSymbolFont{cyrletters}{OT1}{wncyi}{m}{it}
\DeclareSymbolFontAlphabet{\cyrmath}{cyrletters}
\DeclareMathSymbol{\re}{\cyrmath}{cyrletters}{"03}
\newcommand{\Ev}{\re}

\begin{document}
\title{Integrability of Kupershmidt deformations\thanks{This work was
supported in part by NWO-RFBR grant 047017015 (PK, IK, and AV),
RFBR-E.I.N.S.T.E.IN Consortium grant 06-01-92060 (IK, AV, and RV), and
RFBR-CRNS grant 08-07-92496 (IK and AV).}}
\author{P.H.M. Kersten \and I.S. Krasil\cprime shchik \and \\
  A.M. Verbovetsky \and R. Vitolo}
\institute{P.H.M. Kersten \at University of Twente, Postbus 217, 7500 AE
  Enschede, the Netherlands\\
  \email{kerstenphm@ewi.utwente.nl}
  \and
  I.S. Krasil\cprime shchik, A.M. Verbovetsky \at
  Independent University of Moscow, B. Vlasevsky 11, 119002
  Moscow, Russia\\
  \email{josephk@diffiety.ac.ru, verbovet@mccme.ru}
  \and
  R. Vitolo \at Dept.\ of Mathematics ``E. De Giorgi'', Universit\`a del
  Salento, via per Arnesano, 73100 Lecce, Italy\\
  \email{raffaele.vitolo@unile.it} }
\titlerunning{Integrability of Kupershmidt deformations}
\authorrunning{P.H.M. Kersten, I.S. Krasil\cprime shchik,
  A.M. Verbovetsky, R. Vitolo}
\date{\ }

\maketitle

\begin{abstract}
  We prove that the Kupershmidt deformation of a bi-Hamiltonian system is
  itself bi-Hamiltonian.  Moreover, Magri hierarchies of the initial system
  give rise to Magri hierarchies of Kupershmidt deformations as well.  Since
  Kupershmidt deformations are not written in evolution form, we start with an
  outline a geometric framework to study Hamiltonian properties of general
  non-evolution differential equations, developed
  in~\cite{IgoninKerstenKrasilshchikVerbovetskyVitolo:VBGPDE} (see
  also~\cite{KerstenKrasilshchikVerbovetskyVitolo:HSGP}).  \keywords{Nonlinear
    differential equations, variational Schouten bracket, Hamiltonian
    structures, symmetries, conservation laws} \subclass{37K05;35Q53}
\end{abstract}

\section{Introduction}

Recently, in the paper
\cite{Karasu(Kalkani)KarasuSakovichSakovichTurhan:ANIGKdV} the authors derived
the new equation
\begin{equation}\label{eq:1}
 (\partial_x^3 + 8u_x\partial_x + 4u_{xx} )(u_t + u_{xxx} + 6u_x^2) = 0,
\end{equation}
called the KdV6 equation, which turned out to pass the Painlev\'e test.  The
authors introduced the new variables $v=u_x$ and $w=u_t + u_{xxx} + 6u_x^2$,
transforming~\eqref{eq:1} into the system
\begin{equation}
  \label{eq:2}
    v_t + v_{xxx} + 12v v_x - w_x = 0,\quad w_{xxx} + 8v w_x + 4w v_x = 0,
\end{equation}
and posed the problem to study conservation laws and Hamiltonian structures
of the above system.

In \cite{Kupershmidt:KdV6}, Kupershmidt proved the existence of an infinite
series of conservation laws for~\eqref{eq:2} in the following way.  He noted
that~\eqref{eq:2} can be written in the form
\begin{equation}
  \label{eq:3}
      F-A_1(w)=0, \quad A_2(w)=0,
\end{equation}
where $F=0$ is the KdV equation in the unknown $v$, and $A_1$, $A_2$ are the
two standard Hamiltonian operators for $F=0$.  So, \eqref{eq:3} can be
considered to be a deformation of~$F=0$; we call it the \emph{Kupershmidt
deformation}.  The following result yields an infinite series of conservation
laws for~\eqref{eq:3}.

{\def\thetheorem{}
\begin{theorem}[Kupershmidt]
  Let $F=u_t-f=0$ be an evolution bi-Hamiltonian system\textup{,} with
  $A_1$\textup{,}~$A_2$ being the corresponding Hamiltonian operators.  If
  this equation has a Magri hierarchy of conserved densities
  $\frac{dH_i}{dt}=0$\textup{,} $A_1\left(\frac{\delta H_i}{\delta u}\right)
  =A_2\left(\frac{\delta H_{i+1}}{\delta u}\right)$\textup{,} then
  $H_1$\textup{,}~$H_2$\textup{,}~\dots\ are conserved densities
  for~\eqref{eq:3}.
\end{theorem}}
\addtocounter{theorem}{-1}
\begin{proof}
  \begin{multline*}
    \frac{dH_i}{dt} =\left\langle\frac{\delta H_i}{\delta
        u},f+A_1(w)\right\rangle =\left\langle -A_1\left(\frac{\delta
          H_i}{\delta u}\right),
      w\right\rangle \\
    =\left\langle -A_2\left(\frac{\delta H_{i+1}}{\delta u}\right),
      w\right\rangle =\left\langle\frac{\delta H_{i+1}}{\delta
        u},A_2(w)\right\rangle=0.\quad\qed
  \end{multline*}
\end{proof}

Further, Kupershmidt conjectured that these conservation laws commute (in a
sense) so that \eqref{eq:3} is integrable.

What makes this conjecture especially interesting is that system~\eqref{eq:3}
is not in evolution form.  Recently,
in~\cite{IgoninKerstenKrasilshchikVerbovetskyVitolo:VBGPDE} (see
also~\cite{KerstenKrasilshchikVerbovetskyVitolo:HSGP}) we, together with
S.~Igonin, have introduced a generalization of Hamiltonian formalism for
general, not necessarily evolution, systems.

In the present paper, we apply this formalism to the Kupershmidt conjecture.
Namely we prove that the Kupershmidt deformation of every bi-Hamiltonian
equation is again a bi-Hamiltonian system and that every hierarchy of
conservation laws of the original bi-Hamiltonian system gives rise to a
hierarchy of conservation laws of the Kupershmidt deformation.

\section{Preliminaries}
\label{sec:preliminaries}

An adequate setting for dealing with symmetries and conservation laws of
differential equations is provided by jet bundles. Our main sources are
\cite{KrasilshchikVinogradov:SCLDEqMP,KrasilshchikVerbovetsky:HMEqMP};
we will shortly describe our notation below.

From now on all manifolds and maps are $C^\infty$.

\paragraph{Jets and differential equations.}

Let $\pi\colon E\to M$ be a vector bundle. We denote by $J^k(\pi)$ the
corresponding jet manifold and by $\pi_{k,l}\colon J^k(\pi)\to J^l(\pi)$ for
$k>l$ and $\pi_k\colon J^k(\pi)\to E$ the standard projections. The inverse
limit of the chain of projections
$\cdots\to\pi_{k+1,k}\to\pi_{k,k-1}\to\cdots$ is said to be the
\emph{infinite order jet space} and is denoted by $J^\infty(\pi)$.

Let $x^1,\dots,x^n$ be local coordinates in~$M$, $u^1,\dots,u^{m}$ and be local
fiber coordinates in $E$. Then $u^j_\sigma$, where $\sigma$ is a multiindex of
arbitrary length, denote local derivative coordinates on the fibers
of~$J^\infty(\pi)$. If $s\colon M\to E$ is a section then its prolongation
$j_\infty s\colon M\to J^\infty(\pi)$ fulfills $u^j_\sigma\circ j_\infty s =
\pd{^{|\sigma|} (u^j\circ s)}{x^\sigma}$.

We denote by $\mathcal{F}(\pi)$ the algebra of smooth functions on
$J^\infty(\pi)$. This is defined as the direct limit of the chain of inclusions
of smooth functions on $J^k(\pi)$ into smooth functions on $J^{k+1}(\pi)$ via
pull-back. The $\mathcal{F}(\pi)$-module $\Lambda^*(\pi)$ of differential forms
on $J^\infty(\pi)$ is defined in the similar way.

A \emph{horizontal module} is the $\mathcal{F}(\pi)$-module of sections of
$\pi^*_\infty(\alpha)$, where $\alpha$ is a vector bundle over $M$. Denote by
$\kappa(\pi)$ the horizontal module corresponding to the bundle $\pi$
itself.

Let $P_1(\pi)$ and $P_2(\pi)$ be horizontal modules. A linear differential
operator $\Delta\colon P_1(\pi)\to P_2(\pi)$ is called
\emph{$\mathcal{C}$-differential} if it can be restricted to the graphs of all
infinitely prolonged sections of the bundle.  The set of all
$\mathcal{C}$-differential operators from~$P_1(\pi)$ to~$P_2(\pi)$ is denoted
by $\CDiff(P_1(\pi),P_2(\pi))$.  In coordinates, $\mathcal{C}$-differential
operators have the form of a matrix $(a^{\sigma}_{ij}D_{\sigma})$, where
$a^{\sigma}_{ij}\in\mathcal{F}(\pi)$, $D_{\sigma}=D_{i_1}\circ\dots\circ
D_{i_r}$ for $\sigma=i_1\dotsc i_r$ and
$D_i=\pd{}{x^i}+u^j_{\sigma i}\pd{}{u^j_{\sigma}}$
is the total derivative operator
with respect to $x^i$\footnote{The Einstein summation convention will be used
throughout the paper.}.

A $\pi_\infty$-vertical vector field on $J^\infty(\pi)$ is called an
\emph{evolutionary field} if it commutes with all $D_i$ (this property does not
depend on the choice of coordinates).  In coordinates, each evolutionary field
is of the form $\Ev_\phi=D_\sigma(\phi^j)\pd{}{u_\sigma^j}$, where
$\phi^j\in\mathcal{F}(\pi)$. It can be proved that each evolutionary vector
field is uniquely determined by its \emph{generating function}
$\phi\in\kappa(\pi)$ and vice verse to any $\phi\in\kappa(\pi)$
there corresponds an evolutionary field~$\Ev_{\phi}$.

Let $P(\pi)$ be a horizontal module.  For each element $p\in P(\pi)$ there is a
$\mathcal{C}$-differential operator $\ell_p\colon\kappa(\pi)\to P(\pi)$
called the \emph{universal linearization} of~$p$ and defined by
$\ell_p(\phi)=\Ev_\phi(p)$, with $\phi\in\kappa(\pi)$. In
coordinates, $\ell_p$ is the matrix of the form
$(\pd{p^i}{u^j_\sigma}\cdot D_\sigma)$.

A differential form $\omega\in\Lambda^q(\pi)$ on $J^\infty(\pi)$ is called a
\emph{Cartan form} if its pull-back through any prolonged section vanishes.  In
coordinates, Cartan forms contain factors of the type
$\omega_\sigma^j=du_\sigma^j-u_{\sigma i}^j\,dx^i$.  Denote the module of all
Cartan $q$-forms by $\mathcal{C}\Lambda^q(\pi)$.  It is not difficult to show
that $d(\mathcal{C}\Lambda^q(\pi))\subset\mathcal{C}\Lambda^{q+1}(\pi)$.
Therefore
the quotient $\bar{d}$ of $d$, acting on
$\bar{\Lambda}^*(\pi)=\Lambda^*(\pi)/\mathcal{C}\Lambda^*(\pi)$, is well
defined.  The module $\bar{\Lambda}^*(\pi)$ is identified with the submodule
in $\Lambda^*(\pi)$ generated by $\pi^*_{\infty}\Lambda^*(M)$.
Elements of $\bar{\Lambda}^q(\pi)$ are called \emph{horizontal
forms}.  In coordinates, $\bar{\Lambda}^q(\pi)$ is generated by
$f\,dx^{i_1}\wedge\dots\wedge dx^{i_q}$, where $f\in\mathcal{F}(\pi)$, and
$\bar{d}(f\,dx^{i_1}\wedge\dots\wedge dx^{i_q}) =D_i(f)\,dx^i\wedge
dx^{i_1}\wedge\dots\wedge dx^{i_q}$.  The cohomology $\bar{H}^i(\pi)$ of the
complex $(\bar{\Lambda}^*(\pi),\bar{d})$ is said to be the \emph{horizontal
  cohomology}, and coincides with the de~Rham cohomology of $M$ for all degrees
$i$ up to $n-1$.

If $P(\pi)$ is an $\mathcal{F}(\pi)$-module, we write $\hat
P(\pi)=\Hom_{\mathcal{F}(\pi)}(P(\pi),\bar{\Lambda}^n(\pi))$ and consider the
natural pairing $\langle\cdot\,,\cdot\rangle\colon \hat P(\pi)\times
P(\pi)\to\bar{\Lambda}^n(\pi)$.  We recall that for each operator
$\Delta\in\CDiff(P_1(\pi),P_2(\pi))$ there exists a unique operator
$\Delta^*\in\CDiff(\hat P_2(\pi),\hat P_1(\pi))$ such that
\begin{equation*}
  [\langle\hat p_2,\Delta(p_1)\rangle]
  =[\langle\Delta^*(\hat p_2),p_1\rangle],
  \qquad\hat p_2\in\hat P_2(\pi),\quad p_1\in P_1(\pi),
\end{equation*}
where $[\omega]$ denotes the horizontal cohomology class of
$\omega\in\bar{\Lambda}^n$.  The operator~$\Delta^*$ is called
\emph{adjoint} to~$\Delta$.  In coordinates, $(a_{ij}^\sigma D_\sigma)^*
=((-1)^{\abs{\sigma}}D_\sigma\circ
a_{ji}^\sigma)$, where $a_{ij}^\sigma\in\mathcal{F}(\pi)$.

A \emph{differential equation} is a submanifold $\mathcal{E}\subset
J^k(\pi)$. We assume that $\mathcal{E}=\{F=0\}$, where $F\in\pi_k^*(\alpha)$ is
a section of the pull-back bundle of a vector bundle $\alpha$ on $M$.

Let us set $P(\pi)=\pi_\infty^*(\alpha)$.  The equation $\mathcal{E}$ can be
prolonged to a submanifold of $J^\infty(\pi)$, that we still denote by
$\mathcal{E}$. In local coordinates, we have $\mathcal{E}=\{D_\sigma F^k=0\}$
for such a prolongation.
Throughout the paper we will assume equations to be \emph{regular}: this means
that any function $f\in\mathcal{F}(\pi)$ vanishing on $\mathcal{E}$ can be
expressed as $f=\Delta(F)$, where $\Delta\in\CDiff(P(\pi),\mathcal{F}(\pi))$.

\paragraph{Symmetries.}

For any horizontal module $Q(\pi)$ we will denote by $Q$ the restriction of
$Q(\pi)$ to $\mathcal{E}$.  We denote by $\ell_{\mathcal{E}}\colon \kappa \to
P$ the restriction of the universal linearization~$\ell_F$
to $\mathcal{E}$. The space
$\ker \ell_{\mathcal{E}}$ coincides with the space $\Sym\mathcal{E}$ of
\emph{symmetries} of the differential equation $\mathcal{E}$.  The space $\ker
\ell_{\mathcal{E}}^*$ is the space of \emph{cosymmetries}, and is denoted by
$\CoSym\mathcal{E}$.

It can be proved that there exists a sequence
\begin{equation}
  \label{eq:35}
  0 \to \CDiff(P,\mathcal{F}) \to \CDiff(\kappa,\mathcal{F}) \to
  \mathcal{C}\Lambda^1 \to 0.
\end{equation}
Here the first nontrivial map is defined by $\Delta \mapsto
\Delta\circ\ell_{\mathcal{E}}$ and the second one is a natural projection
(see \cite{KrasilshchikVinogradov:SCLDEqMP,KrasilshchikVerbovetsky:HMEqMP} for
details). The differential equation $\mathcal{E}$ is said to be
\emph{normal} if the above sequence is exact. In other words, $\mathcal{E}$ is
normal, if the equality
$\Delta\circ\ell_{\mathcal{E}}$ implies $\Delta=0$.  Most differential
equations of mathematical physics fulfill this property; however, the gauge
equations do not.

\paragraph{Conservation laws.}

Consider the horizontal de Rham complex $(\bar\Lambda^q,\bar d)$ on
$\mathcal{E}$. A \emph{conservation law} is an element $[\omega]\in\bar
H^{n-1}/H^{n-1}(M)
$ (we quotient out the topological conservation laws).
The Vinogradov $\mathcal{C}$-spectral
sequence, see
\cite{KrasilshchikVinogradov:SCLDEqMP,KrasilshchikVerbovetsky:HMEqMP,Vinogradov:CAnPDEqSC}
for a detailed treatment) yields the complex $(E_1^{*,n-1},d_1)$ which plays
the same role as the de Rham complex on jets. In particular the first two terms
of the complex are $E_1^{0,n-1}=\bar H^{n-1}$ and
$E_1^{1,n-1}=\ker\ell_{\mathcal{E}}^*=\CoSym\mathcal{E}$,
and the first differential $d_1\colon
E_1^{0,n-1} \to E_1^{1,n-1}$ is defined by $d_1([\omega]) = \Delta^*(1)$, where
$\Delta\in\CDiff(P,\bar\Lambda^n)$ fulfills $\bar d \omega=\Delta(F)$. The
element $\Delta^*(1)\in\hat P$ is said to be a \emph{generating function} of the
conservation law $[\omega]$.

A $\mathcal{C}$-spectral sequence argument shows that there exists an exact
sequence $0\to H^{n-1}\to \bar H^{n-1} \to \ker\ell_{\mathcal{E}}^*$, where the
last map is just $d_1$. So, the space of conservation laws
$\cl(\mathcal{E})=\bar H^{n-1}/H^{n-1}(M)$ is a subset of the kernel of~
$\ell^*_{\mathcal{E}}$,
$\cl(\mathcal{E})\subset \ker\ell_{\mathcal{E}}^*=\CoSym\mathcal{E}$.

\section{Hamiltonian bivectors on general equations}
\label{sec:hamilt-bivect-gener}

In this section we collect formulas related to the Hamiltonian formalism on
general equations. For details and geometric definitions we refer the reader
to~\cite{IgoninKerstenKrasilshchikVerbovetskyVitolo:VBGPDE} (see
also~\cite{KerstenKrasilshchikVerbovetskyVitolo:HSGP}).

Let $\mathcal{E}\subset J^\infty(\pi)$ be a normal equation given by $F=0$.

A \emph{variational bivector} on~$\mathcal{E}$ is the equivalence class of
$\mathcal{C}$-differential operators $A\colon\hat{P}\to\kappa$
on~$\mathcal{E}$ that satisfy the condition
\begin{equation}\label{eq:4}
  \ell_{\mathcal{E}}\circ A=A^*\circ\ell^*_{\mathcal{E}},
\end{equation}
where two operators are equivalent if they
differ by an operator of the form $\square\circ\ell_{\mathcal{E}}^*$,
$\square=\square^*\colon\hat{\kappa}\to\kappa$.

It is straightforward to see that an action of variational bivectors on
cosymmetries (and, in particular, on conservation laws) is well-defined
and the result is a symmetry.

If $A$ is a bivector then on $J^\infty(\pi)$ we have
\begin{equation}\label{eq:5}
  \ell_F\circ A-A^*\circ\ell_F^*=B(F,\,\cdot\,),
\end{equation}
where $A$ is extended onto~$J^\infty(\pi)$,
$B\colon P(\pi)\times\hat{P}(\pi)\to
P(\pi)$ is a $\mathcal{C}$-differential operator. We denote by
$B^*\colon\hat{P}\times\hat{P}\to\hat{P}$ the operator adjoint to the
operator~$B$ in the first argument and restricted to~$\mathcal{E}$.

Consider an equivalence relation on the set of operators
$\CDiff(\hat{P}\times\hat{P},\hat{P})$ on~$\mathcal{E}$ such that two
operators are equivalent if they differ by an operator of the form
\begin{equation}\label{eq:6}
  \square_1(\ell_{\mathcal{E}}^*(\,\cdot\,),\,\cdot\,)
  +\square_2(\,\cdot\,,\ell_{\mathcal{E}}^*(\,\cdot\,)),
\end{equation}
where $\square_1\colon\hat{\kappa}\times\hat{P}\to\hat{P}$,
$\square_2\colon\hat{P}\times\hat{\kappa}\to\hat{P}$

\begin{proposition}
  For every bivector~$A$ the equivalence class of~$B^*$ is uniquely defined
  and contains a skew-symmetric operator
  $B^*\colon\hat{P}\times\hat{P}\to\hat{P}$.
\end{proposition}

\begin{remark}
  If $\mathcal{E}$ is written in evolution form (in this case $P=\kappa$) then
  we can put $B^*(\psi_1,\psi_2)=\ell^*_{A,\psi_2}(\psi_1)$, here we use the
  notation $\ell_{\Delta,p}=\ell_{\Delta(p)}-\Delta\circ\ell_p$.
  Skew-symmetricity follows from the formula
  \begin{equation}\label{eq:7}
    \ell^*_{\Delta,p}(\hat{q})=\ell^*_{\Delta^*,\hat{q}}(p).
  \end{equation}
\end{remark}

The Schouten bracket of two bivectors is defined by the formula

\begin{equation}\label{eq:8}
  \begin{aligned}
    &[\![A_1,A_2]\!](\psi_1,\psi_2) \\
    &\qquad=\ell_{A_1,\psi_1}(A_2(\psi_2))
    -\ell_{A_1,\psi_2}(A_2(\psi_1)) \\
    &\qquad+\ell_{A_2,\psi_1}(A_1(\psi_2))
    -\ell_{A_2,\psi_2}(A_1(\psi_1)) \\
    &\qquad\qquad-A_1(B_2^*(\psi_1,\psi_2))-A_2(B_1^*(\psi_1,\psi_2)),
  \end{aligned}
\end{equation}
where $\psi_1,\psi_2\in\hat{P}$ and $B_1$,
$B_2$ are the operators which play the role of $B$ in Eq.~(\ref{eq:5}) for
$A_1$, $A_2$, respectively.

The Schouten bracket is a bracket between multivectors.  This means that it is
defined on a wide class of operators. In particular, $0$-vectors are
conservation laws and the bracket between a variational bivector $A$
and a conservation law $\omega$ has the form
\begin{equation*}
  [\![A,\omega]\!]=A(\psi),
\end{equation*}
where $\psi=d_1[\omega]$ is the generating function of $\omega$;
$1$-vectors are symmetries, the bracket between a
symmetry $\phi$ and a conservation law $[\omega]$:
\begin{equation*}
  [\![\phi,\omega]\!]=L_{\phi}(\omega)=\Ev_{\phi}(\psi)+\Delta^*(\psi),
\end{equation*}
where $L_{\phi}$ is the Lie derivative and the operator $\Delta\colon P\to P$ is
defined by the relation $\ell_F(\phi)=\Delta(F)$ on~$J^\infty$
(for equations in evolution form one can take $\Delta=\ell_{\phi}$);
the bracket between two symmetries
\begin{equation*}
  [\![\phi_1,\phi_2]\!]=[\phi_1,\phi_2]\quad\text{(the usual commutator)}.
\end{equation*}
and the bracket between two conservation laws
\begin{equation*}
  [\![\omega_1,\omega_2]\!]=0.
\end{equation*}

A bivector~$A$ is called \emph{Hamiltonian} if $[\![A,A]\!]=0$.

A Hamiltonian bivector~$A$ on~$\mathcal{E}$ gives rise to a Lie algebra
structure on the space of conservation laws of~$\mathcal{E}$:
\begin{equation*}
  \{\omega_1,\omega_2\}_A=[\![[\![A,\omega_1]\!],\omega_2]\!].
\end{equation*}

An equation~$\mathcal{E}$ is called \emph{bi-Hamiltonian} if it has two
Hamiltonian bivectors~$A_1$ and~$A_2$ such that $[\![A_1,A_2]\!]=0$.

\emph{Magri hierarchy} on a bi-Hamiltonian equation~$\mathcal{E}$ is an
infinite sequence $\omega_1$,~$\omega_2$,~\dots{} of conservation laws
of~$\mathcal{E}$ such that $A_1(\psi_i)=A_2(\psi_{i+1})$.

\begin{proposition}
  For Magri hierarchy we get
  \begin{gather*}
    \{\omega_i,\omega_j\}_{A_1}=0, \quad
    \{\omega_i,\omega_j\}_{A_2}=0
    \\
    \{\phi_i,\phi_j\}=0,
  \end{gather*}
  where $\phi_i=A_1(\psi_i)=A_2(\psi_{i+1})$ are symmetries and the bracket
  between them is the commutator:
  $\Ev_{\{\phi_i,\phi_j\}}=[\Ev_{\phi_i},\Ev_{\phi_2}]$.
\end{proposition}

\begin{remark}
  The reader is invited to check that, if $\mathcal{E}$ has an evolutionary
  form, then the above Hamiltonian formalism reduces to the usual one.
\end{remark}

\section{The Kupershmidt deformation}
\label{sec:kupershm-deform}

Let $\mathcal{E}$ be a bi-Hamiltonian equation with Hamiltonian
operators~$A_1$ and~$A_2$ given by $F=0$ as above.  Let us consider the bundle
$\hat{\pi}\colon\hat{E}=E^*\otimes_M\Lambda^n(T^*M)\to M$,
where $E^*\to M$ is the dual bundle to $\pi\colon E\to M$.  We denote by
$w=(w^1,\dots,w^m)$ fiber coordinates on $\hat{\pi}$.

\begin{definition}
  The \emph{Kupershmidt deformation}~$\tilde{\mathcal{E}}\subset
  J^\infty(\pi)\times J^\infty(\hat{\pi})$ has the form
\begin{equation}
  \label{eq:9}
  F+A_1^*(w)=0,\qquad A_2^*(w)=0.
\end{equation}
\end{definition}
We will write $\tilde{F}=(F+A_1^*(w), A_2^*(w))\in \tilde{P} = P\oplus P$,
so that $\tilde{\mathcal{E}}$ is given by $\tilde{F}=0$.

The linearization $\ell_{\tilde{\mathcal{E}}} \colon
\kappa\oplus\hat{\kappa} \to P\oplus P$
and its adjoint have the form
\begin{equation}\label{eq:10}
  \ell_{\tilde{\mathcal{E}}}=
  \begin{pmatrix}
    \ell_{F+A^*_1(w)} & A_1^* \\
    \ell_{A_2^*(w)} & A_2^*
  \end{pmatrix},
  \quad
  \ell^*_{\tilde{\mathcal{E}}}=
  \begin{pmatrix}
    \ell^*_{F+A^*_1(w)} & \ell^*_{A_2^*(w)} \\
    A_1 & A_2
  \end{pmatrix}
\end{equation}
The linearizations in the left- and right-hand sides of these formulas
have different meaning: the left-hand ones are usual
linearizations on~$\tilde{\mathcal{E}}$,
while the right-hand ones are linearizations with respect to the dependent
variables $u$ only, that is, linearizations on~$J^{\infty}(\pi)$.

In what follow we use the following notational rule: linearization of
something marked with tilde is the
$J^\infty(\pi)\times J^\infty(\hat{\pi})$-linearization
(i.e., the linearization  with respect to the dependent variables $u$,~$w$),
otherwise it is the $J^{\infty}(\pi)$-linearization.

Below we often use the obvious relation $\ell_{A^*,w}=\ell_{A^*(w)}$,
which is true since the linearization here is the
$J^{\infty}(\pi)$-linearization.

\begin{lemma}
  \label{sec:kupershm-deform-2}
  We have the following equalities:
  \begin{subequations}\label{eq:11}
   \begin{align}
     \label{eq:12}
    &\ell_{F+A_1^*(w)}\circ A_1 - A_1^*\circ\ell_{F+A_1^*(w)}^* =
    B_1(F+A_1^*(w),\cdot)\\
    \label{eq:13}
    &\ell_{F+A_1^*(w)}\circ A_2 - A_1^*\circ\ell^*_{A_2^*(w)}
    -A_2^*\circ\ell_{F+A_1^*(w)}^* - \ell_{A_2^*(w)}\circ A_1 \notag\\
    &\hphantom{\ell_{F+A_1^*(w)}\circ A_2 - A_1^*\circ\ell^*_{A_2^*(w)}
    -A_2^*\circ\ell_{F+A_1^*(w)}^*}=B_2(F+A_1^*(w),\cdot) + B_1(A_2^*(w),\cdot) \\
    \label{eq:14}
    &\ell_{A_2^*(w)}\circ A_2 - A_2^*\circ\ell_{A_2^*(w)}^*=
    B_2(A_2^*(w),\cdot).
  \end{align}
\end{subequations}
\end{lemma}
\begin{proof}
  We will only prove the statement for the first operator.  For every
  $\psi_1$,~$\psi_2\in\hat{P}$ we have the following equalities  modulo
  $\im\hd$:
  \begin{align*}
    \langle & \ell_F(A_1(\psi_1))+\ell_{A_1^*,w}(A_1(\psi_1)),\psi_2\rangle
    -\langle\ell_F(A_1(\psi_2))+\ell_{A_1^*,w}(A_1(\psi_2)),\psi_1\rangle \\
    & =\langle(\ell_F\circ A_1-A_1^*\circ\ell_F^*)(\psi_1),\psi_2\rangle
    +\langle A_1(\psi_1),\ell^*_{A_1^*,w}(\psi_2)\rangle-
    \langle A_1(\psi_2),\ell^*_{A_1^*,w}(\psi_1)\rangle \\
    & =\langle B_1(F,\psi_1),\psi_2\rangle
    +\langle A_1(\psi_1),\ell^*_{A_1,\psi_2}(w)\rangle-
    \langle A_1(\psi_2),\ell^*_{A_1,\psi_1}(w)\rangle \\
    & =\langle B_1^*(\psi_2,\psi_1),F\rangle
    +\langle(\ell_{A_1,\psi_2}(A_1(\psi_1))
    -\ell_{A_1,\psi_1}(A_1(\psi_2))),w\rangle \\
    & =\langle B_1^*(\psi_2,\psi_1),F\rangle
    +\langle A_1(B_1^*(\psi_2,\psi_1)),w\rangle \\
    & =\langle B_1^*(\psi_2,\psi_1),F+A_1^*(w)\rangle \\
    & =\langle B_1(F+A_1^*(w),\psi_1),\psi_2\rangle.
  \end{align*}

  The equalities in the statement turn out to be true in view of the fact that
  both sides of the above equalities are linear (zero order) operators in
  $\psi_2$, and the only linear operator which is the composition of $\bar{d}$
  and an operator in $\psi_2$ is the zero operator.

  The other two equalities are proved similarly. \qed
\end{proof}

Of course the right-hand sides of the equations in
Lemma~\ref{sec:kupershm-deform-2} vanish on $\tilde{\mathcal{E}}$, hence the
following corollary.

\begin{corollary}\label{sec:kupershm-deform-5}
  The following operators are selfadjoint on $\tilde{\mathcal{E}}$:
  \begin{align*}
    &\ell_F\circ A_1+\ell_{A_1^*(w)}\circ A_1 \\
    &\ell_F\circ A_2+\ell_{A_1^*(w)}\circ A_2-A_1^*\circ\ell^*_{A_2^*,w} \\
    &\ell_{A_2^*(w)}\circ A_2.
  \end{align*}
\end{corollary}

In our computations we will also need the linearization of
Eqs.~(\ref{eq:11}).
\begin{lemma}
  The linearization of Eqs.~(\ref{eq:11}) computed in $\psi_1\in\hat{P}$
  yields the equations
  \begin{subequations}\label{eq:15}
  \begin{align}
    \label{eq:16}
    \begin{split}
    & \ell_{\ell_{F+A_1^*(w)},A_1(\psi_1)}+\ell_{F+A_1^*(w)}\circ\ell_{A_1,\psi_1}
    -\ell_{A_1^*,\ell_{F+A_1^*(w)}^*(\psi_1)}-A_1^*\circ\ell_{\ell_{F+A_1^*(w)}^*,\psi_1}
    \\
    &\hphantom{\ell_{\ell_{F+A_1^*(w)},A_1(\psi_1)}+\ell_{F+A_1^*(w)}\circ\ell_{A_1,\psi_1}
      -\ell_{A_1^*,\ell_{F+A_1^*(w)}^*(\psi_1)}}
    =B_1(\cdot,\psi_1)\circ\ell_{F+A_1^*(w)}
    \end{split}
    \\
    \label{eq:17}
    \begin{split}
    & \ell_{\ell_{F+A_1^*(w)},A_2(\psi_1)} + \ell_{F+A_1^*(w)}\circ\ell_{A_2,\psi_1}
    -\ell_{A_1^*,\ell^*_{A_2^*(w)}(\psi_1)} -
    A_1^*\circ\ell_{\ell_{A_2^*(w)}^*,\psi_1} \\
    &\hphantom{\ell_{\ell_{F+A_1^*(w)},A_2(\psi_1)}}
    -\ell_{A_2^*,\ell_{F+A_1^*(w)}^*(\psi_1)} - A_2^*\circ\ell_{\ell_{F+A_1^*(w)}^*,\psi_1}
    +\ell_{\ell_{A_2^*(w)},A_1(\psi_1)} + \ell_{A_2^*(w)}\circ\ell_{A_1,\psi_1} \\
    &\hphantom{\ell_{\ell_{F+A_1^*(w),A_2(\psi_1)}}}
    =B_2(\cdot,\psi_1)\circ\ell_{F+A_1^*(w)} +
    B_1(\cdot,\psi_1)\circ\ell_{A_2^*(w)}
    \end{split}
    \\
    \label{eq:18}
    & \ell_{\ell_{A_2^*(w)},A_2(\psi_1)}+\ell_{A_2^*(w)}\circ\ell_{A_2,\psi_1}
    -\ell_{A_2^*,\ell_{A_2^*(w)}^*(\psi_1)}-A_2^*\circ\ell_{\ell_{A_2^*(w)}^*,\psi_1}
    =B_2(\cdot,\psi_1)\circ\ell_{A_2^*(w)}
  \end{align}
  \end{subequations}
up to terms which vanish on $\tilde{\mathcal{E}}$.
\end{lemma}
\begin{proof}
  We make use of the formula
  \begin{displaymath}
    \ell_{\Delta\circ\square(\alpha)} = \ell_{\Delta,\square(\alpha)} +
    \Delta\circ\ell_{\square,\alpha} +\Delta\circ\square\circ\ell_\alpha.
  \end{displaymath}
  Let us apply the above formula to Eq.~\eqref{eq:12} computed at $\psi_1$:
  \begin{multline*}
    \ell_{\ell_{F+A_1^*(w)},A_1(\psi_1)}+\ell_{F+A_1^*(w)}\circ\ell_{A_1,\psi_1}
    +\ell_{F+A_1^*(w)}\circ A_1\circ\ell_{\psi_1} \\
    -\ell_{A_1^*,\ell_{F+A_1^*(w)}^*(\psi_1)}-A_1^*\circ\ell_{\ell_{F+A_1^*(w)}^*,\psi_1}
    -A_1^*\circ\ell_{F+A_1^*(w)}^*\circ\ell_{\psi_1}
    \\
    =\ell_{B_1(\cdot,\cdot),(F+A_1^*(w),\psi_1)} +
    B_1(\cdot,\psi_1)\circ\ell_{F+A_1^*(w)}
    +B_1(F+A_1^*(w),\cdot)\circ\ell_{\psi_1}.
  \end{multline*}
  We obtain Eq.~\eqref{eq:16} by observing that the term
  $\ell_{B_1(\cdot,\cdot),(F+A_1^*(w),\psi_1)}$ vanishes on $\tilde{\mathcal{E}}$
  and that terms composed with $\ell_{\psi_1}$ cancel by virtue of
  Eq.~\eqref{eq:12}. The two remaining equations can be derived with
  similar reasoning.\qed
\end{proof}

Now let us consider the operators
$\tilde{A}_1,~\tilde{A}_2\colon \hat{P}\oplus \hat{P} \to
\kappa\oplus\hat{\kappa}$ defined by
\begin{equation}\label{eq:19}
    \tilde{A}_1=
  \begin{pmatrix}
    A_1 & -A_1 \\
    0 & \ell_{F+A_1^*(w)+A_2^*(w)}^*
  \end{pmatrix},
  \qquad
  \tilde{A}_2=
  \begin{pmatrix}
    A_2 & -A_2 \\
    -\ell_{F+A_1^*(w)+A_2^*(w)}^* & 0
  \end{pmatrix}.
\end{equation}

\begin{proposition}
  \label{sec:kupershm-deform-4}
    We have the equalities
    \begin{gather}
      \label{eq:20}
      \ell_{\tilde{F}}\circ\tilde{A}_1
      - \tilde{A}_1^*\circ\ell^*_{\tilde{F}}
      = \tilde{B}_1(\tilde{F},\cdot) \\
      \label{eq:21}
      \ell_{\tilde{F}}\circ\tilde{A}_2
      - \tilde{A}_2^*\circ\ell^*_{\tilde{F}}
      = \tilde{B}_2(\tilde{F},\cdot)
    \end{gather}
    where
    \begin{align}
      & \tilde{B}_1((p_1,p_1'),(\psi_2,\psi_2'))=\notag \\
      &\hphantom{\tilde{B}_1}
       (B_1(p_1,\psi_2)-B_1(p_1,\psi_2'),
    -B_1(p_1,\psi_2)-B_2(p_1,\psi_2')
    -B_1(p_1',\psi_2')-B_2(p_1',\psi_2'))\label{eq:22} \\
      & \tilde{B}_2((p_1,p_1'),(\psi_2,\psi_2'))=\notag \\
      &\hphantom{\tilde{B}_2}
       (B_1(p_1,\psi_2)+B_2(p_1,\psi_2)
    +B_1(p_1',\psi_2)+B_2(p_1',\psi_2'),
    B_2(p_1',\psi_2)-B_2(p_1',\psi_2'))\label{eq:23}
    \end{align}
    It follows that the operators $\tilde{A}_1$, $\tilde{A}_2$ define two
    variational bivectors on $\tilde{\mathcal{E}}$.
\end{proposition}
\begin{proof}
  We have
  \begin{align*}
    \ell_{\tilde{F}} & \circ\tilde{A}_1
      - \tilde{A}_1^*\circ\ell^*_{\tilde{F}} \\
      & =
      \begin{pmatrix}
    \ell_{F+A^*_1(w)} & A_1^* \\
    \ell_{A_2^*(w)} & A_2^*
  \end{pmatrix}\cdot
  \begin{pmatrix}
    A_1 & -A_1 \\
    0 & \ell_{F+A_1^*(w)+A_2^*(w)}^*
  \end{pmatrix} \\
  & \hphantom{\begin{pmatrix}
    \ell_{F+A^*_1(w)} & A_1^* \\
    \ell_{A_2^*(w)} & A_2^*
  \end{pmatrix}\cdot}
  -
  \begin{pmatrix}
    A_1^* & 0 \\
    -A_1^* & \ell_{F+A_1^*(w)+A_2^*(w)}
  \end{pmatrix}
  \cdot \begin{pmatrix}
    \ell^*_{F+A^*_1(w)} & \ell^*_{A_2^*(w)} \\
    A_1 & A_2
  \end{pmatrix} \\
  & =
  \begin{pmatrix}
    B_1(F+A_1^*(w),\cdot) & - B_1(F+A_1^*(w),\cdot) \\
    - B_1(F+A_1^*(w),\cdot) & - B_2(F+A_1^*(w),\cdot) -(B_1+B_2)(A_2^*(w),\cdot)
  \end{pmatrix},
  \end{align*}
  where the last equality is obtained by Lemma~\ref{sec:kupershm-deform-2}.
Eqs.~(\ref{eq:21}) and~(\ref{eq:23}) can be proved in the same way.

The last statement follows directly from the definitions (Eqs.~(\ref{eq:4})
and~(\ref{eq:5})).\qed
\end{proof}

Our next task is to prove that the bivectors from the above proposition endow
$\tilde{\mathcal{E}}$ with a bi-Hamiltonian structure. To this aim we have to
compute the Schouten brackets $[\![\tilde{A}_i,\tilde{A}_j]\!]$ for $i$, $j=1$,
$2$. To do this, we have to compute the linearization
$\ell_{\tilde{A}_i,(\psi_1,\psi_1')}$ and the operators $\tilde{B}_i^*$, with
$i=1$, $2$ (here the adjoint is taken with respect to the first argument). We
have
\begin{equation}\label{eq:24}
  \ell_{\tilde{A}_1,(\psi_1,\psi_1')}
  =
  \begin{pmatrix}
    \ell_{A_1,\psi_1}-\ell_{A_1,\psi_1'} & 0 \\
    \ell_{\ell^*_{F+A_1^*(w)+A^*_2(w)},\psi_1'} & \ell^*_{A_1+A_2,\psi_1'}
  \end{pmatrix},
\end{equation}
and
\begin{equation}\label{eq:25}
    \ell_{\tilde{A}_2,(\psi_1,\psi_1')}
  =
  \begin{pmatrix}
    \ell_{A_2,\psi_1}-\ell_{A_2,\psi_1'} & 0 \\
    -\ell_{\ell^*_{F+A_1^*(w)+A^*_2(w)},\psi_1} & -\ell^*_{A_1+A_2,\psi_1}
  \end{pmatrix}
\end{equation}
directly from the definitions.

\begin{lemma}\label{sec:kupershm-deform-7}
  We have
   \begin{align}
      & \tilde{B}_1^*((\psi_1,\psi_1'),(\psi_2,\psi_2'))\notag \\
      &\hphantom{\tilde{B}_1^*}=
      (B_1^*(\psi_1,\psi_2)-B_1^*(\psi_1,\psi_2')
       -B_1^*(\psi_1',\psi_2)-B_2^*(\psi_1',\psi_2'),
       -B_1^*(\psi_1',\psi_2')-B_2^*(\psi_1',\psi_2'))
       \label{eq:26} \\
      & \tilde{B}_2^*((\psi_1,\psi_1'),(\psi_2,\psi_2'))\notag \\
      &\hphantom{\tilde{B}_2^*}=
      (B_1^*(\psi_1,\psi_2)+B_2^*(\psi_1,\psi_2),
       B_1^*(\psi_1,\psi_2)+B_2^*(\psi_1,\psi_2')
       +B_2^*(\psi_1',\psi_2)-B_2^*(\psi_1',\psi_2'))
    \label{eq:27}
    \end{align}
\end{lemma}
\begin{proof}
  In fact we can write
  \begin{displaymath}
    \tilde{B}_1((p_1,p_1'),(\psi_2,\psi_2'))
    =
    \begin{pmatrix}
      B_1(\cdot,\psi_2)-B_1(\cdot,\psi_2') & 0 \\
      -B_1(\cdot,\psi_2)-B_2(\cdot,\psi_2') &
    -B_1(\cdot,\psi_2')-B_2(\cdot,\psi_2')
    \end{pmatrix}
    \cdot
    \begin{pmatrix}
      p_1 \\ p_1'
    \end{pmatrix}.
  \end{displaymath}
  It follows that
  \begin{displaymath}
    \tilde{B}_1^*((\psi_1,\psi_1'),(\psi_2,\psi_2'))
    =
    \begin{pmatrix}
      B_1^*(\cdot,\psi_2)-B_1^*(\cdot,\psi_2') &
       -B_1^*(\cdot,\psi_2)-B_2^*(\cdot,\psi_2') \\
       0 & -B_1^*(\cdot,\psi_2')-B_2^*(\cdot,\psi_2')
    \end{pmatrix}
    \cdot
    \begin{pmatrix}
      \psi_1 \\ \psi_1'
    \end{pmatrix}.
  \end{displaymath}
  The other identity follows in an analogous way.\qed
\end{proof}

Note that $\tilde{B}_1^*$ and $\tilde{B}_2^*$ are skew-symmetric with respect
to the interchange of the arguments $(\psi_1,\psi_1')$ and $(\psi_2,\psi_2')$.

\begin{theorem}
  \label{sec:kupershm-deform-3}
  The Kupershmidt deformation~$\tilde{\mathcal{E}}$ is a bi-Hamiltonian
  equation with respect to the variational bivectors $\tilde{A}_1$,
  $\tilde{A}_2$.
\end{theorem}

\begin{proof}
  We have to prove the conditions
  \begin{equation}\label{eq:28}
      \begin{aligned}
    &[\![\tilde{A}_i,\tilde{A}_j]\!]((\psi_1,\psi_1'),(\psi_2,\psi_2'))
    \\
    &\qquad=\ell_{\tilde{A}_i,(\psi_1,\psi_1')}(\tilde{A}_j(\psi_2,\psi_2'))
    -\ell_{\tilde{A}_i,(\psi_2,\psi_2')}(\tilde{A}_j(\psi_1,\psi_1'))
    \\
    &\qquad+\ell_{\tilde{A}_j,(\psi_1,\psi_1')}(\tilde{A}_i(\psi_2,\psi_2'))
    -\ell_{\tilde{A}_j,(\psi_2,\psi_2')}(\tilde{A}_i(\psi_1,\psi_1'))
    \\
    &\qquad-\tilde{A}_i(\tilde{B}_j^*((\psi_1,\psi_1'),(\psi_2,\psi_2')))
    -\tilde{A}_j(\tilde{B}_i^*((\psi_1,\psi_1'),(\psi_2,\psi_2')))
    =0,
  \end{aligned}
  \end{equation}
for $i$, $j=1$, $2$. We will only prove them for $i=1$, $j=2$, other
computations being very similar. First of all we compute
the summands in the expression~(\ref{eq:28}):
 \begin{align*}
    & \ell_{\tilde{A}_1,(\psi_1,\psi_1')}(\tilde{A}_2(\psi_2,\psi_2')) \\
    & \hphantom{\ell}
      = \Bigl((\ell_{A_1,\psi_1}-\ell_{A_1,\psi_1'})(A_2(\psi_2)-A_2(\psi_2')), \\
    & \hphantom{ = ()()}
      \doubell(\psi_1,A_2(\psi_2)-A_2(\psi_2'))
      + \ell^*_{A_1+A_2,\psi_1}(-\ell_{F+A_1^*(w)+A_2^*(w)}^*(\psi_2'))\Bigr),
      \\
    & \ell_{\tilde{A}_2,(\psi_1,\psi_1')}(\tilde{A}_1(\psi_2,\psi_2')) \\
    & \hphantom{\ell}
      = \Bigl((\ell_{A_2,\psi_1}-\ell_{A_2,\psi_1'})(A_1(\psi_2)-A_1(\psi_2')), \\
    & \hphantom{ = ()()}
      -\doubell(\psi_1,A_1(\psi_2)-A_1(\psi_2'))
      - \ell^*_{A_1+A_2,\psi_1}(\ell_{F+A_1^*(w)+A_2^*(w)}^*(\psi_2'))\Bigr), \\
    & \tilde{A}_1(\tilde{B}_2^*((\psi_1,\psi_1'),(\psi_2,\psi_2'))) \\
    & = (A_1(B_2^*(\psi_1,\psi_2)-B_2^*(\psi_1,\psi_2')
            -B_2^*(\psi_1',\psi_2)+B_2^*(\psi_1',\psi_2')), \\
    & \hphantom{= ()}
      \ell_{F+A_1^*(w)+A_2^*(w)}^*(B_1^*(\psi_1,\psi_2)+B_2^*(\psi_1,\psi_2')
      +B_2^*(\psi_1',\psi_2)-B_2^*(\psi_1',\psi_2'))), \\
    & \tilde{A}_2(\tilde{B}_1^*((\psi_1,\psi_1'),(\psi_2,\psi_2'))) \\
    & = (A_2(B_1^*(\psi_1,\psi_2)-B_1^*(\psi_1,\psi_2')
            -B_1^*(\psi_1',\psi_2)+B_1^*(\psi_1',\psi_2')), \\
    & \hphantom{= ()}
      -\ell_{F+A_1^*(w)+A_2^*(w)}^*(B_1^*(\psi_1,\psi_2)-B_1^*(\psi_1,\psi_2')
      -B_1^*(\psi_1',\psi_2)-B_2^*(\psi_1',\psi_2'))),
  \end{align*}
  were we introduced the notation
  $\doubell(\psi_1,\phi_2)=\ell_{\ell^*_{F+A_1^*(w)+A^*_2(w)},\psi_1}(\phi_2)$.
  Let us set
  \begin{displaymath}
    (\phi_3,\psi_3') =
    [\![\tilde{A}_1,\tilde{A}_2]\!]((\psi_1,\psi_1'),(\psi_2,\psi_2')).
  \end{displaymath}
  We have
  \begin{displaymath}
    \phi_3=[\![A_1,A_2]\!](\psi_1,\psi_2)+[\![A_1,A_2]\!](\psi_1',\psi_2')
    -[\![A_1,A_2]\!](\psi_1',\psi_2)-[\![A_1,A_2]\!](\psi_1,\psi_2')=0
  \end{displaymath}
  because $[\![A_1,A_2]\!]=0$.

  As for the second component $\psi_3'$, we first observe that the operator
  \begin{equation}\label{eq:29}
    \doubell(\psi_1',\cdot) \colon \kappa\to\hat{\kappa}
  \end{equation}
  is selfadjoint.  Now, we take the adjoint of
  Eqs.~\eqref{eq:15} and compute them in $\psi_2\in\hat{P}$ using
  Eq.~\eqref{eq:7}. We obtain
  \begin{subequations}\label{eq:30}
  \begin{align}
    \label{eq:31}
    \begin{split}
    & \ell_{\ell_{F+A_1^*(w)}^*,\psi_2}^*(A_1(\psi_1))
    +\ell_{A_1,\psi_1}^*(\ell_{F+A_1^*(w)}^*(\psi_2)) \\
    & -\ell_{A_1,\psi_2}^*(\ell_{F+A_1^*(w)}^*(\psi_1))
    -\ell_{\ell_{F+A_1^*(w)}^*,\psi_1}^*(A_1(\psi_2))
      =\ell_{F+A_1^*(w)}^*(-B_1^*(\psi_1,\psi_2))
    \end{split}
    \\
    \label{eq:32}
    \begin{split}
    & \ell_{\ell_{F+A_1^*(w)}^*,\psi_2}^*(A_2(\psi_1))
    + \ell_{A_2,\psi_1}^*(\ell_{F+A_1^*(w)}^*(\psi_2))
    -\ell_{A_1,\psi_2}^*(\ell^*_{A_2^*(w)}(\psi_1))
    - \ell_{\ell_{A_2^*(w)}^*,\psi_1}^*(A_1(\psi_2)) \\
    & -\ell_{A_2,\psi_2}^*(\ell_{F+A_1^*(w)}^*(\psi_1))
    - \ell_{\ell_{F+A_1^*(w)}^*,\psi_1}^*(A_2(\psi_2))
    + \ell_{\ell_{A_2^*(w)}^*,\psi_2}^*(A_1(\psi_1))
    + \ell_{A_1,\psi_1}^*(\ell_{A_2^*(w)}^*(\psi_2)) \\
    &\hphantom{-\ell_{A_2,\psi_2}^*(\ell_{F+A_1^*(w)}^*(\psi_1))}
    =\ell_{F+A_1^*(w)}^*(-B_2^*(\psi_1,\psi_2))
    +\ell_{A_2^*(w)}^*(-B_1(\psi_1,\psi_2)),
    \end{split}
    \\
    \label{eq:33}
    \begin{split}
    & \ell_{\ell_{A_2^*(w)}^*,\psi_2}^*(A_2(\psi_1))
    +\ell_{A_2,\psi_1}^*(\ell_{A_2^*(w)}^*(\psi_2))
    -\ell_{A_2,\psi_2}^*(\ell_{A_2^*(w)}^*(\psi_1))
    -\ell_{\ell_{A_2^*(w)}^*,\psi_1}^*(A_2(\psi_2)) \\
    &\hphantom{\ell_{\ell_{A_2^*(w)}^*,\psi_2}^*(A_2(\psi_1))}
    =\ell_{A_2^*(w)}^*(-B_2^*(\psi_1,\psi_2)).
  \end{split}
\end{align}
  \end{subequations}

  Then, we sum the above three equations to get one single equation, that we
  compute two times in $(\psi_1,\psi_2')$ and $(\psi_1',\psi_2)$
  respectively. By replacing the result into $\psi_3'$ we obtain
  \begin{align*}
    &\psi_3' = \doubell(\psi_2',A_1(\psi_1)+A_2(\psi_1'))
    +\doubell(\psi_2,A_1(\psi_1)+A_2(\psi_1')) \\
    &\hphantom{ \psi_3' =}
    -\doubell(\psi_1',A_1(\psi_2)+A_2(\psi_2'))
    -\doubell(\psi_1,A_1(\psi_2)+A_2(\psi_2')).
  \end{align*}
  We define the operators
  \begin{displaymath}
    \square_1,\square_1'\colon(\hat{\kappa}\times\kappa)
  \times(\hat{P}\times\hat{P})\to(\kappa\times\hat{\kappa}),\quad
  \square_2,\square_2'\colon(\hat{P}\times\hat{P})
  \times(\hat{\kappa}\times\kappa)\to(\kappa\times\hat{\kappa}),
  \end{displaymath}
  as follows:
  \begin{align*}\label{eq:34}
  & \square_1((\phi_1,\phi_1'),(\psi_2,\psi_2'))=
  (0,-\doubell(\psi_2',\phi_1')),\quad
  \square_2((\psi_1,\psi_1'),(\phi_2,\phi_2'))=
  (0,\doubell(\psi_1',\phi_2')),
  \\
  & \square_1'((\phi_1,\phi_1'),(\psi_2,\psi_2'))=
  (0,\doubell(\psi_2,\phi_1')),\quad
  \square_2'((\psi_1,\psi_1'),(\phi_2,\phi_2'))=
  (0,+\doubell(\psi_1,\phi_2')).
  \end{align*}
  We have
  \begin{align*}
    (0,\psi_3') =
    \square_1(\ell_{\tilde{\mathcal{E}}}^*(\psi_1,\psi_1'),(\psi_2,\psi_2')) +
    \square_2((\psi_1,\psi_1'),\ell_{\tilde{\mathcal{E}}}^*(\psi_2,\psi_2')) \\
    \hphantom{(0,\psi_3') =}
    +\square_1'(\ell_{\tilde{\mathcal{E}}}^*(\psi_1,\psi_1'),(\psi_2,\psi_2')) +
    \square_2'((\psi_1,\psi_1'),\ell_{\tilde{\mathcal{E}}}^*(\psi_2,\psi_2')),
  \end{align*}
hence the three-vector $[\![\tilde{A}_1,\tilde{A}_2]\!]$ is zero up
to trivial terms.\qed
\end{proof}

\begin{theorem}
  \label{sec:kupershm-deform-1}
  Let $\omega_1$,~$\omega_2$,~\dots{} is a Magri hierarchy
  for~$\mathcal{E}$.  Suppose that there exist extensions of $\omega_i$, $A_1$,
  and $A_2$ onto~$J^\infty$ such that on~$J^\infty$ we have
  \begin{equation*}
      \begin{gathered}
        \hd\omega_i=\langle\psi_i,F\rangle, \\
        A_1(\psi_i)=A_2(\psi_{i+1}),
      \end{gathered}
    \end{equation*}
  here $\psi_i$ is an extension of the generating function of~$\omega_i$.
  Then $(\psi_i,-\psi_{i+1})$,
  $i=1$,~$2$,~\dots is a Magri hierarchy for the Kupershmidt
  deformation~$\tilde{\mathcal{E}}$.
\end{theorem}

\begin{proof}
  On~$\tilde{J}^\infty$ we have
  \begin{multline*}
    \hd\omega_i=\langle\psi_i,F\rangle
    =\langle\psi_i,F+A_1^*(w)\rangle-\langle\psi_i,A_1^*(w)\rangle \\
    =\langle\psi_i,F+A_1^*(w)\rangle-\langle A_1(\psi_i),w\rangle+\hd\chi_1 \\
    =\langle\psi_i,F+A_1^*(w)\rangle-\langle
    A_2(\psi_{i+1}),w\rangle+\hd\chi_1 \\
    =\langle\psi_i,F+A_1^*(w)\rangle-\langle
    \psi_{i+1},A_2^*(w)\rangle+\hd\chi_2.
  \end{multline*}
  Thus, the form $\omega_i-\chi_2$ is a conservation law with the generating
  function $(\psi_i,-\psi_{i+1})$.  The condition
  $\tilde{A}_1(\psi_i,-\psi_{i+1})=\tilde{A}_2(\psi_{i+1},-\psi_{i+2})$ can be
  easily checked by direct computation.\qed
\end{proof}

\begin{acknowledgements}
We wish to thank Sergey Igonin for reading the manuscript and useful comments.
\end{acknowledgements}

\end{document}